\documentclass[conference]{IEEEtran}

\usepackage{mathtools}
\usepackage{amssymb}
\usepackage{amsmath}
\usepackage{amsthm}
\usepackage{algorithm}
\usepackage{algpseudocode}
\usepackage{enumerate,multirow}
\usepackage{algpseudocode}
\usepackage{tikz}
\usepackage{tikz}
\usetikzlibrary{patterns,positioning,arrows,calc}
\usepackage{siunitx}
\usepackage{multirow}
\usepackage{multicol}
\usepackage{color}
\usepackage{xcolor}
\usepackage{algorithm}
\usepackage{algpseudocode}
\usepackage{graphicx}
\usepackage{subcaption}
\usepackage{hvlogos}
\usepackage{qtree}
\usepackage{tikz}
\usepackage{multirow}
\usepackage{hhline}
\usepackage{hyperref}
\usetikzlibrary{positioning, backgrounds, fit}
\usepackage{tikz}
\usetikzlibrary{positioning, backgrounds, fit}
\usepackage{subcaption}
\usepackage{xcolor}
\usepackage{pgfplots}
\pgfplotsset{compat=1.18}
\usetikzlibrary{pgfplots.statistics}
\algrenewcommand\algorithmicrequire{\textbf{Input:}}
\algrenewcommand\algorithmicensure{\textbf{Output:}}
\usepackage{comment}

\usepackage{geometry}
\usepackage{tikz}
\usepackage{tikz}
\usetikzlibrary{patterns,positioning,arrows,calc}\usepackage{graphicx}
\usepackage{graphicx}
\usepackage{tikz}
\usetikzlibrary{trees}
\usepackage{amsmath}

\date{}

\usepackage{calc}
\newcolumntype{M}[1]{>{\centering\arraybackslash}m{#1}}
\newcolumntype{N}{@{}m{0pt}@{}}

\usepackage{mdwmath}
\usepackage{blindtext}
\usepackage{eqparbox}
\usepackage{stfloats}
\usepackage[numbers,sort&compress]{natbib}
\usepackage{algorithm}
\usepackage{algpseudocode}
\usepackage{tikz}
\usetikzlibrary{patterns, arrows.meta, positioning}
\IEEEoverridecommandlockouts
\title{\Huge
{Multiple-Bases Belief Propagation List Decoding for Quantum LDPC Codes}}
\author{\IEEEauthorblockN{Sheida Rabeti and Hessam Mahdavifar} 
\IEEEauthorblockA{Department of Electrical and Computer Engineering, Northeastern University, Boston, MA 02115, USA \\ 
Email: \{rabeti.s, h.mahdavifar\}@northeastern.edu}
\thanks{This work was supported in part by NSF under Grant CCF-2415440.}
}

\newtheorem{theorem}{Theorem}

\AtBeginEnvironment{remark}{\itshape}
\newtheorem{lemma}[theorem]{Lemma}

\newtheorem{example}{Example}

\newcommand{\cG}{{\cal G}}

\newcommand{\cL}{{\cal L}}

\newcommand{\cP}{{\cal P}}

\newcommand{\cS}{{\cal S}}

\DeclareMathAlphabet{\mathbfsl}{OT1}{ppl}{b}{it} 

\newcommand{\be}[1]{\begin{equation}\label{#1}}
\newcommand{\ee}{\end{equation}}

\renewcommand{\le}{\leqslant} 
\renewcommand{\leq}{\leqslant}
\renewcommand{\ge}{\geqslant} 
\renewcommand{\geq}{\geqslant}

\newcommand{\Cref}[1]{Co\-ro\-lla\-ry\,\ref{#1}}

\begin{document}
\vspace{10mm}
\maketitle
\begin{abstract}
In this paper, we propose a belief-propagation (BP)-based decoder, termed the Multiple-Bases Belief-Propagation List Decoder (MBBP-LD), for quantum low-density parity-check (QLDPC) codes. The key idea is to generate \emph{structured decoding diversity} by constructing multiple redundant parity-check representations via cycle-free subtree decompositions of the Tanner graph, and running BP decoding in parallel across these representations. This extends the classical Multiple-Bases Belief-Propagation (MBBP) framework to the quantum setting while preserving the linear-time complexity and efficiency of standard BP decoding, and avoids the need for super-linear post-processing.

Simulation results demonstrate that MBBP-LD improves upon existing BP-based decoders, including BP with ordered statistics decoding (BP-OSD) and belief propagation with guided decimation (BPGD) across several QLDPC codes, while requiring substantially fewer total BP iterations. For bivariate bicycle codes $[[144,12,12]]$ and $[[288,12,18]]$, MBBP-LD achieves up to $20\%$ reduction in error rate compared to BPGD and up to $30\%$ compared to BP-OSD in the low- and moderate-error regimes. For the larger B1 code $[[882, 24, 18 \leq d \leq 24]]$, MBBP-LD attains comparable or improved performance relative to BPGD while maintaining BP-like decoding latency under parallel implementation.
\end{abstract}

\section{Introduction}

Quantum error-correcting codes are essential for fault-tolerant quantum computation. Among existing constructions, quantum low-density parity-check (QLDPC) codes have attracted significant attention due to their sparse parity-check matrices, which reduce qubit--qubit interactions during syndrome extraction and make them well suited for architectures with limited qubit connectivity. Building on the success of classical LDPC codes introduced by Gallager \cite{gallager1962low}, recent advances in QLDPC code constructions have further highlighted their potential for scalable quantum computing \cite{7336474,Breuckmann_2021,vasic2025quantumlowdensityparitycheckcodes}. However, realizing this potential in practice depends on the availability of efficient and high-performance decoding algorithms.

Belief propagation (BP) remains one of the most widely used decoding algorithms for both classical and quantum LDPC codes due to its low complexity and highly optimized hardware implementations. However, in QLDPC codes, its performance is often degraded by code degeneracy, short cycles, and trapping sets in the Tanner graph. A variety of BP-based improvements have been proposed to address these limitations. These include BP with ordered statistics decoding (BP-OSD), which achieves strong performance at the cost of a super-linear complexity dominated by the order-statistics stage scaling as $\mathcal{O}(n^3)$ in the worst case (where $n$ is the block length) \cite{Panteleev_2021}, as well as iterative enhancements such as layered decoding, serial scheduling, and bit-flipping methods \cite{crest2023layereddecodingquantumldpc,Raveendran2021trappingsetsof,chytas2025enhancedminsumdecodingquantum}. Other approaches rely on additional post-processing, including stabilizer inactivation (BP-SI) and belief propagation with guided decimation (BPGD), with complexities $\mathcal{O}(n^2\log n)$ and $\mathcal{O}(n^2)$, respectively \cite{crest2023stabilizerinactivationmessagepassingdecoding,yao2024beliefpropagationdecodingquantum}. Further variants such as Ambiguity Clustering, Localized Statistics Decoding, and the Ordered Tanner Forest method also aim to improve decoding performance through increasingly sophisticated post-processing steps \cite{wolanski2024ambiguity,hillmann2025localized,demarti2024almost}. 

In this work, we propose a Multiple-Bases Belief-Propagation List-Decoding (MBBP-LD) framework for QLDPC codes. The key idea is to generate \emph{structured decoding diversity} within the BP framework, enabling improved error-rate performance while preserving the low-complexity and hardware-efficient nature of BP. In particular, the proposed approach maintains decoding latency comparable to standard BP under parallel implementation and avoids the need for super-linear post-processing, in contrast to existing high-performance BP-based decoders such as BP-OSD and BPGD.

More specifically, the proposed decoder constructs multiple redundant parity-check representations, executes BP decoding on each in parallel, and aggregates the resulting candidate error estimates through a list-based decision stage. Building on classical MBBP decoding \cite{Hehn_2010,4557244}, our approach extends BP-based methods in two important respects. First, we employ augmented parity-check constructions that allow duplicated parity-check rows \cite{rigby2018augmented,rigby2019modified}. Second, and more importantly, instead of random row selection as in \cite{rigby2018augmented,rigby2019modified}, we introduce deterministic subtree-based decompositions of the Tanner graph to construct \emph{structured} redundant representations. These cycle-free substructures induce diverse decoding trajectories across parallel BP instances, thereby mitigating trapping-set effects and enhancing decoding performance.

\section{Preliminaries}
\subsection{Quantum Stabilizer and CSS Codes}
An $[[n, k, d]]$ quantum stabilizer code is a $2^k$-dimensional subspace $\mathcal{C} \subseteq (\mathbb{C}^2)^{\otimes n}$ with an Abelian stabilizer group 
$\mathcal{S}$:
\begin{equation} \label{eq:StabilizerCode} \mathcal{C} = \{|\psi\rangle \in (\mathbb{C}_2)^{\otimes n} \colon s |\psi\rangle= |\psi\rangle, \forall s \in \mathcal{S} \}. \end{equation}
Each generator $g \in \mathcal{S}$ acts as a parity-check constraint. The \textit{minimum distance} $d$ of a stabilizer code is the minimum weight of some Pauli operator $P \in \cP_n$ commuting with elements in $\cS$ such that $P \notin \cS$. 

A Calderbank--Shor--Steane (CSS) code is a stabilizer code with a parity-check matrix of the form
$
H = \begin{bmatrix} H_X & 0 \\ 0 & H_Z \end{bmatrix},
$
where $H_X,H_Z$ are classical binary parity-check matrices satisfying $H_X H_Z^T = 0$. Such codes can correct Pauli-$X$ and Pauli-$Z$ errors independently using $H_Z$ and $H_X$, respectively. In this work, we focus on QLDPC CSS codes, where both $H_X$ and $H_Z$ are sparse.

\subsection{Belief Propagation (BP) Decoding}
BP decoding operates on the Tanner graph of $\cG = (V_c \cup V_v, E)$, where $V_c$ are check nodes, $V_v$ are variable nodes, and $E$ denotes edges. Also, $d_v$ and $d_c$ denote the maximum degrees of variable nodes and check nodes, respectively. For each variable $v_i$, the log-likelihood ratio (LLR) is initialized as $\mu_i = \log \frac{1-p}{p}$, where $p$ is the physical error probability (assumed identical for Pauli-$X$ and Pauli-$Z$ errors). During iterations, messages are updated using the normalized min-sum rule:
\vspace{-2mm}
\begin{align}
m^{(t+1)}_{c \to v} &= \alpha \!\!\left(\prod_{v' \in N(c)\setminus v} \mathrm{sign}\big(m^{(t)}_{v' \to c}\big)\right)
\min_{v' \in N(c)\setminus v} \!\!\left| m^{(t)}_{v' \to c} \right|, \\
m^{(t+1)}_{v \to c} &= \mu_i + \sum_{c' \in N(v)\setminus c} m^{(t)}_{c' \to v}.
\end{align}
The posterior LLR is $m^{(t)}_i = \mu_i + \sum_{c \in N(v_i)} m^{(t)}_{c \to v_i}$.
A hard decision is made as $\hat{e}_i = 0$ if $m^{(t)}_i \ge 0$, otherwise $\hat{e}_i=1$. The algorithm halts when all parity checks satisfy the syndrome (i.e., a valid error pattern is found) or when a maximum number of iterations is reached. Serial scheduling updates check nodes sequentially and immediately propagates messages within the same iteration, improving convergence speed. In QLDPC codes, the presence of short cycles and degeneracy can significantly affect the convergence behavior of BP, motivating the need for modified decoding strategies.

\subsection{Belief Propagation with Augmented Parity Checks}
Augmented parity-check (APC) decoding improves BP by introducing redundancy through duplicating a subset of rows in the original parity-check matrix \cite{rigby2018augmented,rigby2019modified}. Let $H\in\mathbb{F}_2^{m\times n}$ denote the original parity-check matrix and define the augmented matrix as
$
H_A=
\begin{bmatrix}
H\\
H_{\text{dup}}
\end{bmatrix},
$
where $H_{\text{dup}}$ contains duplicated rows from $H$. BP is then executed on $H_A$ instead of $H$. Duplicated checks place more emphasis on certain parity constraints and change the convergence behavior of BP, which can help mitigate trapping sets while preserving the same per-iteration complexity as standard BP \cite{rigby2019modified}. Existing APC methods typically generate $H_{\text{dup}}$ through random row selection. However, the use of randomly selected duplicated rows may not fully exploit the underlying graph structure, which motivates more structured approaches to constructing redundant parity-check representations.

\section{Multiple-Bases Belief-Propagation List-decoding (MBBP-LD)}
\label{sec:RPC-LD}

In this section, we extend the classical Multiple-Bases Belief-Propagation (MBBP) framework \cite{Hehn_2010,4557244} to the quantum setting. Unlike conventional augmented parity-check (APC) decoding \cite{rigby2018augmented,rigby2019modified}, which generates redundant rows through random duplication of parity checks, we construct structured redundant parity-check representations based on subtree decompositions of the Tanner graph. We also introduce a decision rule for selecting the final estimate from the candidate list based on candidate frequency and Hamming weight.

\subsection{MBBP Decoding via Tree-Based Construction}
\label{sec:RPC-LD-A}
MBBP decoding runs multiple BP decoders in parallel, each operating on a distinct parity-check matrix representation of the same code. Let $H^{(\ell)}$, $\ell \in \{1, \ldots, L\}$, denote the matrix used by the $\ell$-th decoder, and let $\hat{e}_\ell$ be the corresponding error estimate after at most $i$ iterations. Each decoder that converges contributes its output to a candidate list $\mathcal{L} = \{\hat{e}_s \mid s \in \mathcal{S}\}$, where $\mathcal{S} \subseteq \{1, \ldots, L\}$ denotes the set of successful decoders. In prior work \cite{Hehn_2010,4557244}, the elements of $\mathcal{L}$ are passed to a least-metric selector (LMS), which selects the most likely codeword according to the channel distribution. In this work, we instead employ an alternative decision rule, described next.

We construct the parity-check matrices used by the parallel decoders by augmenting the original matrix $H$ with redundant layers derived from a collection $\mathcal{T}$ of maximal subtrees in the Tanner graph that partition the check nodes. A subtree $t \in \mathcal{T}$ is called \emph{maximal} if no additional check node, together with its adjacent variable nodes, can be included without forming a cycle. Intuitively, these cycle-free substructures capture locally tree-like components of the Tanner graph, which are known to improve the behavior of BP.

The set $\mathcal{T}$ is obtained by exploring the check nodes and forming maximal subtrees, as detailed in Algorithm \ref{alg:subtree}. Each subtree $t$ defines a submatrix $H_t$, and the corresponding augmented matrix is given by
\begin{equation}
    H^{(t)} = 
    \begin{bmatrix}
        H \\[2pt]
        H_t
    \end{bmatrix}.
\end{equation}
BP decoding is then run in parallel on all constructed matrices $\{H^{(t)} \mid t \in \mathcal{T}\}$. Each decoder produces an estimate $\hat{e}_t$ after a fixed number of iterations, and those that converge contribute their outputs to the candidate list $\mathcal{L} = \{\hat{e}_t \mid t \in \mathcal{T}_{\mathrm{conv}}\}$, where $\mathcal{T}_{\mathrm{conv}} \subseteq \mathcal{T}$ denotes the set of successful decoders. The list $\mathcal{L}$ is then passed to a decision function $f_{\mathrm{DM}}$, which selects the final estimate $\hat{e}$ according to a certain selection rule described in the next subsection. The overall decoding scheme is illustrated in Fig. \ref{fig:tree-mbbp}, with pseudocode provided in Algorithm \ref{alg:mbbp-decoder}.

\begin{figure}[t]
\centering
\begin{tikzpicture}[>=stealth', font=\sffamily]

\node[draw, minimum width=0.9cm, minimum height=2.5cm, label=above:$H$] (H) {};

\node[draw, pattern=north east lines, pattern color=gray,
      minimum width=0.9cm, minimum height=0.5cm,
      anchor=north] (S1) at (H.north) {$t_1$};

\node at (H.center) {$\vdots$};

\node[draw, pattern=horizontal lines, pattern color=gray,
      minimum width=0.9cm, minimum height=0.5cm,
      anchor=south] (Sl) at (H.south) {$t_\ell$};

\draw[->, thick] (H.east) -- ++(0.7,0);
\coordinate (Hmid) at (H.center);

\node[draw, minimum width=0.9cm, minimum height=0.5cm,
      right=1.7cm of Hmid, yshift=1cm, anchor=center] (H1) {$H$};
\node[draw, pattern=north east lines, pattern color=gray,
      minimum width=0.9cm, minimum height=0.6cm,
      below=0cm of H1.south, anchor=north] (HS1) {$H_{t_1}$};

\draw[->] (HS1.north east) -- ++(1.0,0) node[midway, above] {\scriptsize BP-Dec}
           coordinate (bp1);

\node[draw, minimum width=0.9cm, minimum height=0.5cm,
      right=1.7cm of Hmid, yshift=-0.45cm, anchor=center] (H2) {$H$};
\node[draw, pattern=horizontal lines, pattern color=gray,
      minimum width=0.9cm, minimum height=0.3cm,
      below=0cm of H2.south, anchor=north] (HSl) {$H_{t_\ell}$};

\draw[->] (HSl.north east) -- ++(1.0,0) node[midway, above] {\scriptsize BP-Dec}
           coordinate (bpl);

\node[draw, fill=yellow!40, minimum width=1.4cm, minimum height=0.5cm,
      right=4.7cm of Hmid, anchor=center] (L) {$\mathcal{L}$};
\node[draw, fill=orange!40, minimum width=1.4cm, minimum height=0.5cm,
      right=1.7cm of L, anchor=center] (DM) {$f_{\mathrm{DM}}$};

\draw[->, thick] (L.east) -- (DM.west);
\draw[->, thick] (DM.east) -- ++(0.5,0) node[right] {$\hat{e}$};

\draw[densely dashed, ->] (bp1) to[bend left=12] (L.west);
\draw[densely dashed, ->] (bpl) to[bend right=12] (L.west);

\node[font=\small] at ($(bp1)!0.45!(L.west)+(0,0.35)$) {$\hat{e}_1$};
\node[font=\small] at ($(bpl)!0.45!(L.west)-(0,0.35)$) {$\hat{e}_\ell$};

\end{tikzpicture}
\caption{MBBP-LD decoding with redundant-row construction.}
\label{fig:tree-mbbp}
\end{figure}
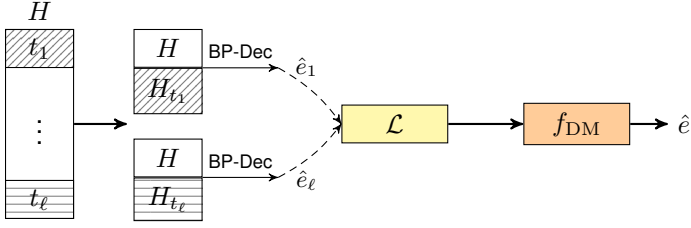

\begin{algorithm}[t]
\caption{MBBP-LD Decoder via Decision-Maker $f_{\mathrm{DM}}$ for QLDPC Codes}
\label{alg:mbbp-decoder}
\begin{algorithmic}[1]
\Require Parity-check matrix $H \in \{0,1\}^{m \times n}$; collection of maximal subtrees $\mathcal{T} = \{t_1, \ldots, t_{|\mathcal{T}|}\}$; syndrome vector $s \in \{0,1\}^m$; channel parameter $p$; decision rule $f_{\mathrm{DM}}$
\Ensure Estimated error vector $\hat{e}$

\State Initialize candidate list $\mathcal{L} \gets \emptyset$

\Statex \textbf{Decoding Phase:}
\For{each $t \in \mathcal{T}$}
    \State $H^{(t)} \gets [\,H; H_t\,]$
    \State $(\hat{e}_t, \textit{converged}) \gets \textsc{BP-Decode}(H^{(t)}, s, p)$
    \If{\textit{converged}}
        \State $\mathcal{L} \gets \mathcal{L} \cup \{\hat{e}_t\}$
    \EndIf
\EndFor

\Statex \textbf{Decision Phase:}
\State $\hat{e} \gets f_{\mathrm{DM}}(\mathcal{L})$

\Return $\hat{e}$
\end{algorithmic}
\end{algorithm}

Algorithm \ref{alg:subtree} partitions the check nodes into a collection of maximal subtrees $\mathcal{T}$ based on a permutation $\pi$ of $V_c$, which determines the order in which check nodes are selected as roots. Starting from each unassigned root, the subtree is expanded by including its neighboring variable nodes and subsequently their adjacent check nodes. A candidate check node is added only if its inclusion preserves the cycle-free property of the induced subgraph. This expansion continues until no further check nodes can be included without creating a cycle, yielding a maximal subtree. Repeating this procedure for all unassigned check nodes produces a partition of $V_c$ into disjoint subsets. Different permutations $\pi$ lead to different collections $\mathcal{T}$ and hence different augmented matrices $H^{(t)}$, resulting in diverse decoding behaviors.

\begin{algorithm}
\caption{Maximal Subtrees Construction}
\label{alg:subtree}
\begin{algorithmic}[1]
\Require Parity-check matrix $H \in \{0,1\}^{m \times n}$; permutation $\pi$ of 
check nodes $V_c$, defining the order in which roots are selected
\Ensure Assignment $\texttt{check\_sets} \in \mathbb{Z}^m$ mapping each check node to its subtree index

\State Construct the Tanner graph $G = (V_c \cup V_v, E)$.
\State Initialize $\texttt{check\_sets}[c] \leftarrow 0$ for all $c \in V_c$; set subtree index $s \leftarrow 0$.

\For{each $c \in V_c$ in order $\pi$}
    \If{$\texttt{check\_sets}[c] = 0$}
        \State $s \leftarrow s + 1$
        \State Initialize BFS queue $Q \leftarrow \{c\}$.
        \State Mark all $v \in V_v$ as unvisited; mark all $c' \in V_c$ as untested.
        \While{$Q \neq \emptyset$}
            \State Dequeue $u$ from $Q$; mark $u$ as tested.
            \State Let $\eta \leftarrow |\{v \in N(u) : v \text{ is visited}\}|$
            \If{$\eta \leq 1$}
                \Comment{At most one visited variable neighbour $\Rightarrow$ cycle-free}
                \State $\texttt{check\_sets}[u] \leftarrow s$
                \For{each $v \in N(u)$}
                    \State Mark $v$ as visited.
                    \For{each untested $c' \in N(v)$ with $\texttt{check\_sets}[c'] = 0$}
                        \State Enqueue $c'$ into $Q$.
                    \EndFor
                \EndFor
            \EndIf
        \EndWhile
    \EndIf
\EndFor
\State \Return $\texttt{check\_sets}$
\end{algorithmic}
\end{algorithm}

\definecolor{varblue}{RGB}{180,210,240}
\definecolor{varborder}{RGB}{50,110,170}
\definecolor{t1fill}{RGB}{240,185,165}
\definecolor{t1border}{RGB}{195,85,45}
\definecolor{t1text}{RGB}{195,85,45}
\definecolor{t2fill}{RGB}{170,220,200}
\definecolor{t2border}{RGB}{25,130,100}
\definecolor{t2text}{RGB}{25,130,100}
\definecolor{rejectgray}{RGB}{200,200,200}
\definecolor{rejectborder}{RGB}{120,120,120}
 
\tikzset{
  vnode/.style={
    circle, draw=varborder, fill=varblue,
    minimum size=16pt, inner sep=0pt, font=\scriptsize
  },
  ct1/.style={
    rectangle, rounded corners=2pt,
    draw=t1border, fill=t1fill,
    minimum size=16pt, inner sep=0pt, font=\scriptsize
  },
  ct2/.style={
    rectangle, rounded corners=2pt,
    draw=t2border, fill=t2fill,
    minimum size=16pt, inner sep=0pt, font=\scriptsize
  },
  crej/.style={
    rectangle, rounded corners=2pt,
    draw=rejectborder, fill=rejectgray,
    minimum size=16pt, inner sep=0pt, font=\scriptsize
  },
  tedge/.style={draw=black!80, line width=0.5pt},
  treeedge/.style={draw=black, line width=0.8pt},
}
 
\begin{figure}[t]
  \centering
  \scalebox{0.77}{%
\begin{tikzpicture}
 

\draw[t1border, dashed, rounded corners=4pt, line width=0.8pt]
  (-0.6,-2.9) rectangle (3.6,6.6);

\draw[t2border, dashed, rounded corners=4pt, line width=0.8pt]
  (3.9,-2.9) rectangle (8.1,6.6);

 
\node[font=\footnotesize] at (3.6, 7.5) {$
H = \left[\begin{array}{cccccc}
{\color{t1text}1} & {\color{t1text}1} & {\color{t1text}1} & {\color{t1text}0} & {\color{t1text}0} & {\color{t1text}0} \\[2pt]
{\color{t2text}0} & {\color{t2text}1} & {\color{t2text}1} & {\color{t2text}0} & {\color{t2text}0} & {\color{t2text}1} \\[2pt]
{\color{t1text}1} & {\color{t1text}0} & {\color{t1text}0} & {\color{t1text}1} & {\color{t1text}1} & {\color{t1text}0} \\[2pt]
{\color{t2text}0} & {\color{t2text}0} & {\color{t2text}0} & {\color{t2text}1} & {\color{t2text}1} & {\color{t2text}1}
\end{array}\right]
\!\!
\begin{array}{l}
{\color{t1text}\leftarrow c_1} \\[2pt]
{\color{t2text}\leftarrow c_2} \\[2pt]
{\color{t1text}\leftarrow c_3} \\[2pt]
{\color{t2text}\leftarrow c_4}
\end{array}
$};
 
\node[font=\small] at (2.22, 8.4) {$v_1$};
\node[font=\small] at (2.72, 8.4) {$v_2$};
\node[font=\small] at (3.22, 8.4) {$v_3$};
\node[font=\small] at (3.72, 8.4) {$v_4$};
\node[font=\small] at (4.22, 8.4) {$v_5$};
\node[font=\small] at (4.72, 8.4) {$v_6$};

 
\node[font=\small\bfseries, text=t1border] at (1.5, 6.0)
  {Subtree $t_1 = \{c_1, c_3\}$};
 
\node[ct1] (c1) at (1.5, 5.3) {$c_1$};
 
\node[vnode] (v1) at (0.5, 4.3) {$v_1$};
\node[vnode] (v2) at (1.5, 4.3) {$v_2$};
\node[vnode] (v3) at (2.5, 4.3) {$v_3$};
 
\draw[treeedge] (c1) -- (v1);
\draw[treeedge] (c1) -- (v2);
\draw[treeedge] (c1) -- (v3);
 
\node[ct1] (c3) at (0.5, 3.3) {$c_3$};
\draw[treeedge] (v1) -- (c3);
 
\node[vnode] (v4) at (0.0, 2.3) {$v_4$};
\node[vnode] (v5) at (1.0, 2.3) {$v_5$};
\draw[treeedge] (c3) -- (v4);
\draw[treeedge] (c3) -- (v5);
 
\node[crej] (rc2) at (2, 3.3) {$c_2$};
\node[font=\small, text=black, align=center] at (2, 2.8)
  {$\eta{=}2$};
\draw[black, dashed, line width=0.8pt] (v2) -- (rc2);
\draw[black, dashed, line width=0.8pt] (v3) -- (rc2);

\node[crej] (rc4) at (0.5, 1.3) {$c_4$};
\node[font=\small, text=black, align=center] at (0.5, 0.8)
  {$\eta{=}2$};
\draw[black, dashed, line width=0.8pt] (v4) -- (rc4);
\draw[black, dashed, line width=0.8pt] (v5) -- (rc4);

 
\def\xoff{4.35}

\node[font=\small\bfseries, text=t2border] at (\xoff+1.5, 6.0)
  {Subtree $t_2 = \{c_4, c_2\}$};
 
\node[ct2] (c4) at (\xoff+1.5, 5.3) {$c_4$};
 
\node[vnode] (sv4) at (\xoff+0.5, 4.3) {$v_4$};
\node[vnode] (sv5) at (\xoff+1.5, 4.3) {$v_5$};
\node[vnode] (sv6) at (\xoff+2.5, 4.3) {$v_6$};
 
\draw[treeedge] (c4) -- (sv4);
\draw[treeedge] (c4) -- (sv5);
\draw[treeedge] (c4) -- (sv6);
 
\node[ct2] (c2) at (\xoff+2.5, 3.3) {$c_2$};
\draw[treeedge] (sv6) -- (c2);
 
\node[vnode] (sv2) at (\xoff+2.0, 2.3) {$v_2$};
\node[vnode] (sv3) at (\xoff+3.0, 2.3) {$v_3$};
\draw[treeedge] (c2) -- (sv2);
\draw[treeedge] (c2) -- (sv3);


\draw[t1border!40, line width=0.4pt] (-0.45,0.5) -- (3.45,0.5);
\draw[t2border!40, line width=0.4pt] (4.05,0.5) -- (7.95,0.5);


\node[font=\small\bfseries, text=t1border] at (1.5,0.25) {$H^{(t_1)}$};
\node[font=\small\bfseries, text=t2border] at (6.0,0.25) {$H^{(t_2)}$};


\node[scale=0.88] at (1.45,-1.35) {$
\left[\begin{array}{cccccc}
1 & 1 & 1 & 0 & 0 & 0 \\[1pt]
0 & 1 & 1 & 0 & 0 & 1 \\[1pt]
1 & 0 & 0 & 1 & 1 & 0 \\[1pt]
0 & 0 & 0 & 1 & 1 & 1 \\[2pt]
\hline\\[-5pt]
{\color{t1text}1} & {\color{t1text}1} & {\color{t1text}1} & {\color{t1text}0} & {\color{t1text}0} & {\color{t1text}0} \\[1pt]
{\color{t1text}1} & {\color{t1text}0} & {\color{t1text}0} & {\color{t1text}1} & {\color{t1text}1} & {\color{t1text}0}
\end{array}\right]
\begin{array}{@{}l@{}}
\left.\rule{0pt}{1.02cm}\right\} H \\[0.02cm]
\left.\rule{0pt}{0.52cm}\right\} {\color{t1text}H_{t_1}}
\end{array}
$};


\node[scale=0.88] at (6.00,-1.35) {$
\left[\begin{array}{cccccc}
1 & 1 & 1 & 0 & 0 & 0 \\[1pt]
0 & 1 & 1 & 0 & 0 & 1 \\[1pt]
1 & 0 & 0 & 1 & 1 & 0 \\[1pt]
0 & 0 & 0 & 1 & 1 & 1 \\[2pt]
\hline\\[-5pt]
{\color{t2text}0} & {\color{t2text}0} & {\color{t2text}0} & {\color{t2text}1} & {\color{t2text}1} & {\color{t2text}1} \\[1pt]
{\color{t2text}0} & {\color{t2text}1} & {\color{t2text}1} & {\color{t2text}0} & {\color{t2text}0} & {\color{t2text}1}
\end{array}\right]
\begin{array}{@{}l@{}}
\left.\rule{0pt}{1.02cm}\right\} H \\[0.02cm]
\left.\rule{0pt}{0.52cm}\right\} {\color{t2text}H_{t_2}}
\end{array}
$};
 
\end{tikzpicture}%
}
\caption{Parity-check matrix $H$ and subtree partition under permutation $\pi=(c_1,c_4,c_2,c_3)$. (a) Matrix $H$. (b) Subtrees $t_1$ and $t_2$. (c) Augmented matrices $H^{(t_1)}$ and $H^{(t_2)}$.}
\label{fig:subtree-partition}
\end{figure}
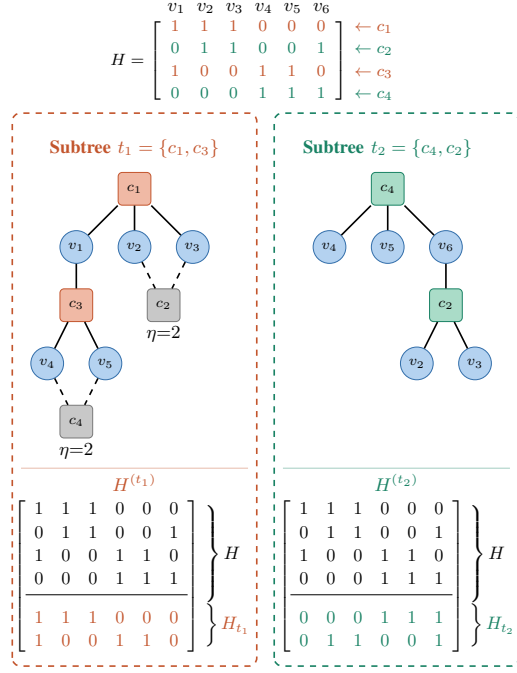

\begin{example}
Fig. \ref{fig:subtree-partition} illustrates the subtree construction for the parity-check $H$ under the permutation $\pi=(c_1,c_4,c_2,c_3)$. The permutation first selects $c_1$ as the root of the first subtree. The algorithm expands through its neighboring variable nodes and adds $c_3$, since it does not add any cycle. Check nodes $c_2$ and $c_4$ are excluded because they would introduce cycles (shown by $\eta=2$). This forms the first subtree $t_1=\{c_1,c_3\}$.
The next unassigned check node in $\pi$ is $c_4$, which becomes the root of the second subtree. Expanding from $c_4$ adds $c_2$ while maintaining the tree structure, resulting in $t_2=\{c_4,c_2\}$. The resulting partition is $\mathcal{T}=\{t_1,t_2\}$, which generates the augmented matrices $H^{(t_1)}$ and $H^{(t_2)}$.
\end{example}

Lemma \ref{lem:subtree-size} provides an upper bound on the number of additional rows introduced by any subtree augmentation, which directly determines the complexity of BP decoding over the augmented parity-check matrices. Its implications for serial-scheduling BP are further discussed in Section \ref{sec:num-res}.
\begin{lemma}[Subtree Size Bound]
\label{lem:subtree-size}
For a Tanner graph $\mathcal{G} = (V_c \cup V_v, E)$ with check-regular degree $w$, 
the number of check nodes in any subtree $t \in \mathcal{T}$ generated by Algorithm  \ref{alg:subtree} satisfies
\begin{equation}
    |t| \le \frac{|V_v| - 1}{w - 1} \triangleq \delta(|V_v|, w).
\end{equation}
In particular, for $w = 6$ generalized bicycle (GB) codes where $|V_v| = 2|V_c|$, this bound reduces to
\begin{equation}
    |t| \le \frac{2|V_c| - 1}{5} \le 0.4\,|V_c|.
\end{equation}
\end{lemma}

\begin{proof}
 Each check node in a tree of degree $w$ is connected to $w$ distinct variable nodes. 
When a new check is added to the subtree, it must share exactly one variable node 
with the existing check nodes to maintain acyclicity; hence, each additional check 
introduces exactly $(w-1)$ new variable nodes. 

Consequently, a subtree containing $|t|$ check nodes covers $1 + (w-1)|t|$ variable nodes considering the check node root as well. 
Since the Tanner graph contains at most $|V_v|$ variable nodes, it follows that
$1 + (w-1)|t| \le |V_v|$, which simplifies to the desired bound $|t| \le (|V_v| - 1)/(w - 1)$. 
\end{proof}



Intuitively, the proposed subtree-based construction improves decoding performance by mitigating trapping-set effects \cite{Raveendran2021trappingsetsof,1003839} while generating structured redundant representations. Since BP is exact on tree graphs \cite{gallager1962low}, decoding over each subtree-induced submatrix $H_t$ benefits from locally cycle-free structure. Moreover, the resulting subtrees yield connected and structured partitions of the Tanner graph, particularly in graphs with many short cycles, which are closely related to girth properties and similar structural observations in \cite{Rabeti2025BoundsAN}. Compared to random augmentation \cite{rigby2019modified}, these structures induce more diverse decoding trajectories and improve convergence behavior.

\subsection{Decision Making (DM) Rule}



After constructing the candidate list $\cL$, the output error vector is selected as $\hat{e} = f_{\mathrm{DM}}(\cL)$. We propose the following rule:

\textbf{Frequency-Weighted Scoring (FWS).}
Each candidate is assigned a score based on its frequency in the list and its Hamming weight:
$
    f_{\mathrm{DM}}^{\text{FWS}}(\cL) 
    = \arg\max_{e \in \cL} \frac{|\{e' \in \cL : e' = e\}|}{w_H(e) + 1}.
$
The numerator favors candidates that appear repeatedly across BP instances, while the denominator penalizes higher-weight errors. If no decoder converges, we declare a decoding failure and return the all-zero error vector.

\subsection{Complexity and Latency Analysis}

In Algorithm \ref{alg:mbbp-decoder}, MBBP-LD runs multiple BP instances in parallel over augmented matrices with at most $\delta$ additional rows, as defined in Lemma \ref{lem:subtree-size}. For check-regular graphs with bounded degree $w$, this yields $\delta = \mathcal{O}(|V_c|)$. In particular, for weight-$6$ generalized bicycle (GB) codes with $|V_v| = 2|V_c|$, we have $\delta \le 0.4\,|V_c|$, so the augmentation increases the number of check nodes by at most a constant factor.

The decoding latency, assuming all BP instances are executed in parallel, is determined by the slowest parallel instance and scales as
$
\mathcal{O}\!\left(
K_{\max}\big(|V_v| d_v + (|V_c| + \delta)d_c\big)
\right),
$
where $K_{\max} \triangleq \max_{t \in \mathcal{T}} K_t$ is the maximum number of BP iterations over all subtrees. For bounded $d_v$ and $d_c$, this reduces to $\mathcal{O}(K_{\max} |V_v|)$, yielding latency on the same order as standard BP decoding.

\section{Numerical Results}
\label{sec:num-res}

In Fig. \ref{fig:fer_loglog144}--\ref{fig:fer_loglogB1}, we evaluate MBBP-LD on two bivariate bicycle codes, $[[144,12,12]]$ and $[[288,12,18]]$ \cite{Bravyi_2024}, and the B1 code $[[882,24,18 \leq d \leq 24]]$ \cite{Panteleev_2021}, over a binary symmetric channel with independent $X$-type errors. Each point is obtained after 100 decoding failures.

We compare MBBP-LD with BP, BP-Serial, BP-OSD \cite{Panteleev_2021}, and BPGD under both parallel \cite{yao2024beliefpropagationdecodingquantum} and serial scheduling \cite{moradi2026sequential}. All methods use a maximum of $I_{\max}=100$ iterations. BP, BP-Serial, and MBBP-LD employ normalized min-sum decoding with scaling factor $\alpha=0.875$, where BP uses parallel scheduling and the others use serial scheduling. BPGD variants use product-sum decoding with $T=100$ decimation steps. The BP, BP-Serial, and BP-OSD implementations follow the libraries in \cite{Roffe_2020, Roffe_LDPC_Python_tools_2022}.

Across all codes, MBBP-LD consistently outperforms BP-OSD and BPGD, while remaining competitive with BPGD-Serial. For the $[[144,12,12]]$ code, it achieves $3$--$24\%$ reduction in logical error rate (LER) compared to BP-OSD and $7$--$17\%$ compared to BPGD, with $2$--$8\%$ gains over BPGD-Serial at moderate-to-high noise. For the $[[288,12,18]]$ code, the reduction reaches up to $49\%$ compared to BP-OSD and $36\%$ compared to BPGD, with up to $10\%$ improvement over BPGD-Serial for $p \geq 0.05$.


For the B1 code $[[882,24]]$, MBBP-LD achieves significant gains over BP-OSD (exceeding $91\%$ at $p=0.04$) and up to $66\%$ reduction in logical error rate (LER) compared to BPGD (parallel) at low noise, while remaining competitive at higher noise levels. Although BPGD-Serial offers slightly better LER in the low-noise regime, this gap must be interpreted in light of the substantially different computational cost. For example, at $p=0.04$, BPGD-Serial requires about $2.5$ iterations on average, compared to $21.7$ for BPGD and $14.4$ for MBBP-LD. This reflects the effect of decimation, which can reduce iterations per stage but incurs significantly higher overall decoding cost. Furthermore, BPGD-based decoders typically employ product-sum updates, whereas MBBP-LD uses normalized min-sum that may offer significantly lower computational complexity.

Table \ref{tab:avg_iters} highlights the corresponding complexity trade-off for the $[[288,12,18]]$ code. MBBP-LD requires only a moderate increase in iterations compared to BP-OSD (e.g., $81.31$ vs.\ $43.03$ at $p=0.08$), remaining within the same order and preserving BP-like latency. In contrast, BPGD and BPGD-Serial incur orders-of-magnitude larger total iteration counts (e.g., exceeding $9.5\times10^3$ at $p=0.08$ and $2\times10^4$ at $p=0.10$) due to repeated decoding and decimation. Thus, while BPGD-based methods can achieve lower LER in the low-noise regime, MBBP-LD offers a favorable trade-off between error-rate performance and decoding efficiency, avoiding both super-linear OSD post-processing and the large iteration overhead of decimation-based decoding.

\definecolor{colorBPOSD}{HTML}{1F77B4}
\definecolor{colorBP}{HTML}{8C8C8C}
\definecolor{colorBPSerial}{HTML}{FF7F0E}
\definecolor{colorBPGD}{HTML}{2CA02C}
\definecolor{colorBPGDSerial}{HTML}{8C564B}
\definecolor{colorMBBPLD}{HTML}{6A0DAD}

\begin{figure}[t]
    \centering
    \resizebox{0.8\columnwidth}{!}{%
    \begin{tikzpicture}
    \begin{loglogaxis}[
        scale only axis,
        width  = 6.5cm,
        height = 6.5cm,
        xlabel = {Physical Error Rate $p_x$},
        ylabel = {Logical Error Rate},
        xlabel style     = {font=\small},
        ylabel style     = {font=\small},
        tick label style = {font=\small},
        xmin = 0.018, xmax = 0.112,
        ymin = 5e-5,  ymax = 1.5,
        xtick       = {0.02, 0.03, 0.04, 0.05, 0.06, 0.07, 0.08, 0.09, 0.10},
        xticklabels = {$0.02$, $0.03$, $0.04$, $0.05$, $0.06$,
                       $0.07$, $0.08$, $0.09$, $0.10$},
        xticklabel style = {rotate=45, anchor=north east, font=\small},
        xtick align      = inside,
        ytick align      = inside,
        enlarge x limits = 0.02,
        enlarge y limits = 0.02,
        clip             = true,
        grid             = both,
        grid style       = {line width=0.4pt, black!25},
        minor grid style = {line width=0.3pt, black!15},
        minor tick num   = 1,
        axis line style  = {line width=0.8pt},
        legend style = {
            at           = {(0.99, 0.01)},
            anchor       = south east,
            font         = \footnotesize,
            fill         = white,
            draw         = black,
            cells        = {anchor=west},
            column sep   = 2pt,
            row sep      = 1pt,
            inner sep    = 4pt,
        },
        legend columns    = 1,
        legend cell align = left,
    ]

    \addplot[color=colorBP, mark=square*, mark size=3pt,
             line width=1.4pt, solid] coordinates {
        (0.02,9.990e-04)(0.03,4.908e-03)(0.04,1.902e-02)(0.05,6.367e-02)
        (0.06,1.412e-01)(0.07,2.526e-01)(0.08,4.197e-01)(0.09,6.055e-01)(0.10,7.821e-01)
    };
    \addlegendentry{BP}

    \addplot[color=colorBPSerial, mark=triangle*, mark size=3.5pt,
             line width=1.4pt, solid] coordinates {
        (0.02,1.250e-04)(0.03,1.620e-03)(0.04,1.160e-02)(0.05,4.401e-02)
        (0.06,1.093e-01)(0.07,2.150e-01)(0.08,3.607e-01)(0.09,5.688e-01)(0.10,7.051e-01)
    };
    \addlegendentry{BP-Serial}

    \addplot[color=colorBPOSD, mark=*, mark size=3pt,
             line width=1.4pt, solid] coordinates {
        (0.02,1.152e-04)(0.03,1.438e-03)(0.04,1.021e-02)(0.05,3.815e-02)
        (0.06,9.708e-02)(0.07,1.942e-01)(0.08,3.404e-01)(0.09,5.321e-01)(0.10,6.603e-01)
    };
    \addlegendentry{BP-OSD}

    \addplot[color=colorBPGD, mark=diamond*, mark size=4pt,
             line width=1.4pt, solid] coordinates {
        (0.02,1.023e-04)(0.03,1.303e-03)(0.04,8.661e-03)(0.05,3.374e-02)
        (0.06,8.403e-02)(0.07,1.831e-01)(0.08,3.090e-01)(0.09,5.000e-01)(0.10,6.731e-01)
    };
    \addlegendentry{BPGD}

    \addplot[color=colorBPGDSerial, mark=*, mark size=3pt,
             line width=1.4pt, dashed] coordinates {
        (0.02,9.156e-05)(0.03,1.136e-03)(0.04,8.043e-03)(0.05,3.081e-02)
        (0.06,8.194e-02)(0.07,1.727e-01)(0.08,3.007e-01)(0.09,5.000e-01)(0.10,6.603e-01)
    };
    \addlegendentry{BPGD-Serial}

    \addplot[color=colorMBBPLD, mark=star, mark size=4.5pt,
             mark options={fill=colorMBBPLD, line width=0.8pt},
             line width=1.4pt, solid] coordinates {
        (0.02,8.989e-05)(0.03,1.082e-03)(0.04,7.733e-03)(0.05,2.934e-02)
        (0.06,7.829e-02)(0.07,1.599e-01)(0.08,2.768e-01)(0.09,4.587e-01)(0.10,6.410e-01)
    };
    \addlegendentry{MBBP-LD (proposed)}

    \end{loglogaxis}
    \end{tikzpicture}}%
    \caption{Decoding performance of the proposed MBBP-LD for Bivariate Bicycle Code $[[144,12,12]]$ \cite{Bravyi_2024}
    }
    \label{fig:fer_loglog144}
\end{figure}
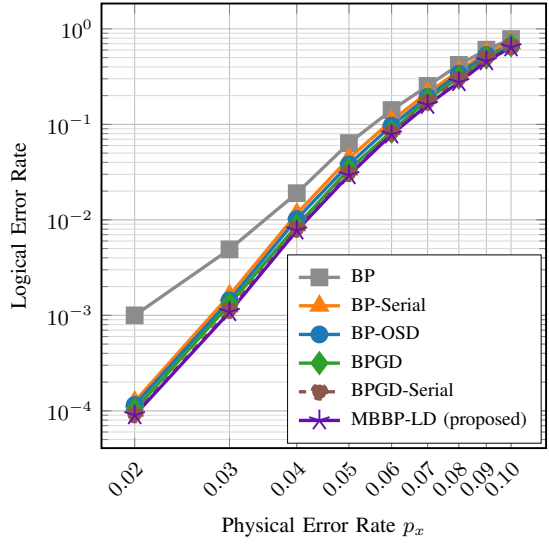


\definecolor{colorBPOSD}{HTML}{1F77B4}
\definecolor{colorBP}{HTML}{8C8C8C}
\definecolor{colorBPSerial}{HTML}{FF7F0E}
\definecolor{colorBPGD}{HTML}{2CA02C}
\definecolor{colorBPGDSerial}{HTML}{8C564B}
\definecolor{colorMBBPLD}{HTML}{6A0DAD}

\begin{figure}[t]
    \centering
    \resizebox{0.8\columnwidth}{!}{%
    \begin{tikzpicture}
    \begin{loglogaxis}[
        scale only axis,        
        width  = 6.5cm,         
        height = 6.5cm,
        xlabel = {Physical Error Rate $p_x$},
        ylabel = {Logical Error Rate},
        xlabel style     = {font=\small},
        ylabel style     = {font=\small},
        tick label style = {font=\small},
        xmin = 0.028, xmax = 0.105,
        ymin = 8e-6,  ymax = 1.5,
        xtick       = {0.03, 0.04, 0.05, 0.06, 0.07, 0.08, 0.09, 0.10},
        xticklabels = {$0.03$, $0.04$, $0.05$, $0.06$,
                       $0.07$, $0.08$, $0.09$, $0.10$},
        xticklabel style = {rotate=45, anchor=north east, font=\small},
        xtick align      = inside,
        ytick align      = inside,
        enlarge x limits = 0.02,
        enlarge y limits = 0.02,
        clip             = true,
        grid             = both,
        grid style       = {line width=0.4pt, black!25},
        minor grid style = {line width=0.3pt, black!15},
        minor tick num   = 1,
        axis line style  = {line width=0.8pt},
        legend style = {
            at           = {(0.99, 0.01)},
            anchor       = south east,
            font         = \footnotesize,
            fill         = white,
            draw         = black,
            cells        = {anchor=west},
            column sep   = 2pt,
            row sep      = 1pt,
            inner sep    = 4pt,
        },
        legend columns    = 1,
        legend cell align = left,
    ]

    \addplot[color=colorBP, mark=square*, mark size=3pt,
             line width=1.4pt, solid] coordinates {
        (0.03,3.155e-03)(0.04,1.062e-02)(0.05,3.712e-02)(0.06,1.083e-01)
        (0.07,2.549e-01)(0.08,4.772e-01)(0.09,6.557e-01)(0.10,7.875e-01)
    };
    \addlegendentry{BP}

    \addplot[color=colorBPSerial, mark=triangle*, mark size=3.5pt,
             line width=1.4pt, solid] coordinates {
        (0.03,6.785e-05)(0.04,1.669e-03)(0.05,1.620e-02)(0.06,6.408e-02)
        (0.07,1.806e-01)(0.08,3.673e-01)(0.09,5.943e-01)(0.10,7.375e-01)
    };
    \addlegendentry{BP-Serial}

    \addplot[color=colorBPOSD, mark=*, mark size=3pt,
             line width=1.4pt, solid] coordinates {
        (0.03,3.768e-05)(0.04,9.820e-04)(0.05,1.122e-02)(0.06,4.609e-02)
        (0.07,1.434e-01)(0.08,3.056e-01)(0.09,5.047e-01)(0.10,6.938e-01)
    };
    \addlegendentry{BP-OSD}

    \addplot[color=colorBPGD, mark=diamond*, mark size=4pt,
             line width=1.4pt, solid] coordinates {
        (0.03,3.124e-05)(0.04,7.750e-04)(0.05,9.031e-03)(0.06,4.167e-02)
        (0.07,1.290e-01)(0.08,3.056e-01)(0.09,5.047e-01)(0.10,6.625e-01)
    };
    \addlegendentry{BPGD}

    \addplot[color=colorBPGDSerial, mark=*, mark size=3pt,
             line width=1.4pt, dashed] coordinates {
        (0.03,1.549e-05)(0.04,4.830e-04)(0.05,7.371e-03)(0.06,3.441e-02)
        (0.07,1.135e-01)(0.08,2.788e-01)(0.09,4.811e-01)(0.10,6.375e-01)
    };
    \addlegendentry{BPGD-Serial}

    \addplot[color=colorMBBPLD, mark=star, mark size=4.5pt,
             mark options={fill=colorMBBPLD, line width=0.8pt},
             line width=1.4pt, solid] coordinates {
        (0.03,1.994e-05)(0.04,5.020e-04)(0.05,6.640e-03)(0.06,3.157e-02)
        (0.07,1.032e-01)(0.08,2.681e-01)(0.09,4.717e-01)(0.10,6.250e-01)
    };
    \addlegendentry{MBBP-LD (proposed)}

    \end{loglogaxis}
    \end{tikzpicture}}
    \caption{Decoding performance of the proposed MBBP-LD for Bivariate Bicycle Code $[[288,12,18]]$ \cite{Bravyi_2024}
    }
    \label{fig:fer_loglog288}
\end{figure}
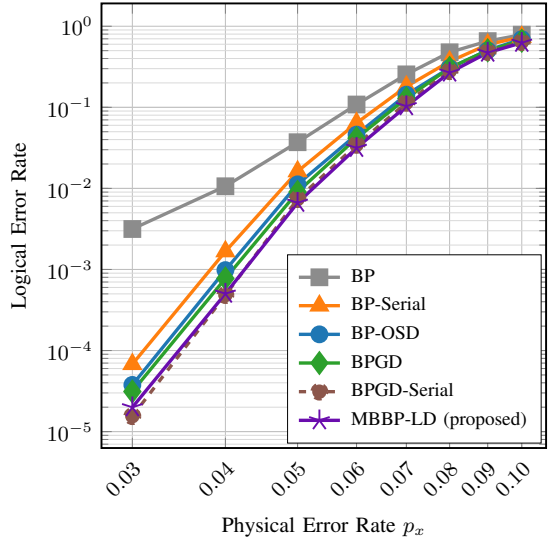

 

\definecolor{colorBPOSD}{HTML}{1F77B4}
\definecolor{colorBP}{HTML}{8C8C8C}
\definecolor{colorBPSerial}{HTML}{FF7F0E}
\definecolor{colorBPGD}{HTML}{2CA02C}
\definecolor{colorBPGDSerial}{HTML}{8C564B}
\definecolor{colorMBBPLD}{HTML}{6A0DAD}

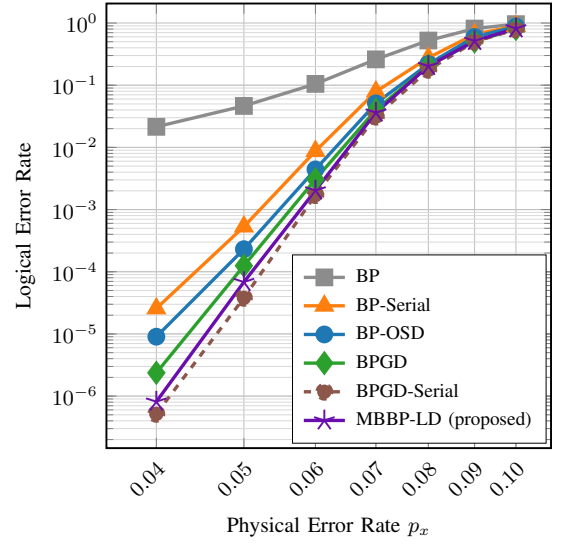
\begin{figure}[t]
    \centering
    \resizebox{0.8\columnwidth}{!}{%
    \begin{tikzpicture}
    \begin{loglogaxis}[
        scale only axis,
        width  = 6.5cm,
        height = 6.5cm,
        xlabel = {Physical Error Rate $p_x$},
        ylabel = {Logical Error Rate},
        xlabel style     = {font=\small},
        ylabel style     = {font=\small},
        tick label style = {font=\small},
        xmin = 0.036, xmax = 0.107,
        ymin = 2e-7,  ymax = 1.5,
        xtick       = {0.04, 0.05, 0.06, 0.07, 0.08, 0.09, 0.10},
        xticklabels = {$0.04$, $0.05$, $0.06$, $0.07$, $0.08$, $0.09$, $0.10$},
        xticklabel style = {rotate=45, anchor=north east, font=\small},
        xtick align      = inside,
        ytick align      = inside,
        enlarge x limits = 0.02,
        enlarge y limits = 0.02,
        clip             = true,
        grid             = both,
        grid style       = {line width=0.4pt, black!25},
        minor grid style = {line width=0.3pt, black!15},
        minor tick num   = 1,
        axis line style  = {line width=0.8pt},
        legend style = {
            at           = {(0.99, 0.01)},
            anchor       = south east,
            font         = \footnotesize,
            fill         = white,
            draw         = black,
            cells        = {anchor=west},
            column sep   = 2pt,
            row sep      = 1pt,
            inner sep    = 4pt,
        },
        legend columns    = 1,
        legend cell align = left,
    ]

    \addplot[color=colorBP, mark=square*, mark size=3pt,
             line width=1.4pt, solid] coordinates {
        (0.04,2.142e-02)(0.05,4.641e-02)(0.06,1.049e-01)(0.07,2.604e-01)
        (0.08,5.233e-01)(0.09,8.144e-01)(0.10,9.590e-01)
    };
    \addlegendentry{BP}

    \addplot[color=colorBPSerial, mark=triangle*, mark size=3.5pt,
             line width=1.4pt, solid] coordinates {
        (0.04,2.567e-05)(0.05,5.317e-04)(0.06,8.730e-03)(0.07,7.896e-02)
        (0.08,2.759e-01)(0.09,6.598e-01)(0.10,8.934e-01)
    };
    \addlegendentry{BP-Serial}

    \addplot[color=colorBPOSD, mark=*, mark size=3pt,
             line width=1.4pt, solid] coordinates {
        (0.04,9.026e-06)(0.05,2.308e-04)(0.06,4.426e-03)(0.07,5.107e-02)
        (0.08,2.191e-01)(0.09,5.979e-01)(0.10,8.770e-01)
    };
    \addlegendentry{BP-OSD}

    \addplot[color=colorBPGD, mark=diamond*, mark size=4pt,
             line width=1.4pt, solid] coordinates {
        (0.04,2.368e-06)(0.05,1.249e-04)(0.06,3.066e-03)(0.07,4.020e-02)
        (0.08,2.049e-01)(0.09,5.052e-01)(0.10,8.033e-01)
    };
    \addlegendentry{BPGD}

    \addplot[color=colorBPGDSerial, mark=*, mark size=3pt,
             line width=1.4pt, dashed] coordinates {
        (0.04,5.179e-07)(0.05,3.801e-05)(0.06,1.705e-03)(0.07,3.115e-02)
        (0.08,1.765e-01)(0.09,5.000e-01)(0.10,7.869e-01)
    };
    \addlegendentry{BPGD-Serial}

    \addplot[color=colorMBBPLD, mark=star, mark size=4.5pt,
             mark options={fill=colorMBBPLD, line width=0.8pt},
             line width=1.4pt, solid] coordinates {
        (0.04,8.139e-07)(0.05,6.833e-05)(0.06,2.030e-03)(0.07,3.622e-02)
        (0.08,2.028e-01)(0.09,5.155e-01)(0.10,8.197e-01)
    };
    \addlegendentry{MBBP-LD (proposed)}

    \end{loglogaxis}
    \end{tikzpicture}}%
    \caption{Decoding performance of the proposed MBBP-LD for B1 code $[[882, 24, 18 \leq d \leq 24]]$ \cite{Panteleev_2021}
    }
    \label{fig:fer_loglogB1}
\end{figure}




 \begin{table}[t]
\centering
\renewcommand{\arraystretch}{1.1}
\begin{tabular}{lcccc}
\hline
Decoder & $p=0.04$ & $p=0.06$ & $p=0.08$ & $p=0.10$ \\
\hline
BP-OSD        & 3.34  & 12.07   & 43.03   & 77.06 \\
MBBP-LD       & 13.79 & 49.48   & 81.31   & 94.94 \\
BPGD          & 40.02 & 1447.73 & 9559.32 & 20432.68 \\
BPGD-Serial   & 19.59 & 1184.65 & 9144.15 & 20209.15 \\
\hline
\end{tabular}
\caption{Average BP iterations of code $[[288,12,18]]$ \cite{Bravyi_2024}; BPGD variants include iterations over all decimation steps.}
\label{tab:avg_iters}
\end{table}
\begin{table}[t!]
\centering

\begin{tabular}{c|cc|cc}
\hline
\multirow{2}{*}{$p$} 
& \multicolumn{2}{c|}{Tree} 
& \multicolumn{2}{c}{Random} \\
\cline{2-5}
& Mean LER & Best LER & Mean LER & Best LER \\
\hline
0.08 & 0.2928 & 0.2690 & 0.3038 & 0.2865 \\
0.07 & 0.1004 & 0.0904 & 0.1030 & 0.0954 \\
0.06 & 0.0289 & 0.0250 & 0.0303 & 0.0273 \\
0.05 & 0.00511 & 0.00454 & 0.00545 & 0.00489 \\
0.04 & 0.000425 & 0.000366 & 0.000469 & 0.000412 \\
\hline
\end{tabular}
\caption{Comparison of subtree-based and matched random partitions for the BB code $[[288,12,18]]$ \cite{Bravyi_2024}.}
\label{tab:tree_vs_random}
\end{table}

\subsection{Comparison with Random Augmentation}

Table \ref{tab:tree_vs_random} compares subtree-based partitions with matched random partitions for the BB code $[[288,12,18]]$ \cite{Bravyi_2024} across different error rates. We consider $20$ partitions in each case. For a fair comparison, random partitions are constructed by replacing subtree augmentations with randomly selected rows while preserving the same group structure and number of duplicated rows. We observe that subtree-based partitions consistently achieve lower LER than their random counterparts, with relative improvements ranging from $2.5\%$ to $9.4\%$, and more pronounced gains in the low-noise regime.

\section{Conclusion}
\label{sec:conclusion}

In this paper, we proposed the MBBP-LD decoder for QLDPC codes, based on structured redundant parity-check representations derived from subtree decompositions of the Tanner graph. The proposed framework improves error-rate performance over existing BP-based decoders in various scenarios, including BP-OSD and BPGD, while maintaining BP-like decoding latency under parallel implementation. Future work includes designing code-specific constructions of redundant checks and developing more refined decision rules. The framework also naturally supports additional decoding diversity through different permutations and layerings with minimal impact on latency.

\bibliographystyle{IEEEtran}
{\footnotesize \bibliography{ref}}

@article{gallager1962low,
  title={Low-density parity-check codes},
  author={Gallager, Robert},
  journal={IRE Transactions on information theory},
  volume={8},
  number={1},
  pages={21--28},
  year={1962},
  publisher={IEEE}
}

@article{7336474,
  title={Fifteen years of quantum {LDPC} coding and improved decoding strategies},
  author={Babar, Zunaira and Botsinis, Panagiotis and Alanis, Dimitrios and Ng, Soon Xin and Hanzo, Lajos},
  journal={iEEE Access},
  volume={3},
  pages={2492--2519},
  year={2015},
  publisher={IEEE}
}

@article{Hehn_2010,
  title={Multiple-bases belief-propagation decoding of high-density cyclic codes},
  author={Hehn, Thorsten and Huber, Johannes B and Milenkovic, Olgica and Laendner, Stefan},
  journal={IEEE transactions on communications},
  volume={58},
  number={1},
  pages={1--8},
  year={2010},
  publisher={IEEE}
}

@inproceedings{4557244,
  title={Multiple-bases belief-propagation for decoding of short block codes},
  author={Hehn, Thorsten and Huber, Johannes B and Laendner, Stefan and Milenkovic, Olgica},
  booktitle={2007 IEEE International Symposium on Information Theory},
  pages={311--315},
  year={2007},
  organization={IEEE}
}

@article{Raveendran2021trappingsetsof,
  title={Trapping sets of quantum {LDPC} codes},
  author={Raveendran, Nithin and Vasi{\'c}, Bane},
  journal={Quantum},
  volume={5},
  pages={562},
  year={2021},
  publisher={Verein zur F{\"o}rderung des Open Access Publizierens in den Quantenwissenschaften}
}

@inproceedings{crest2023layereddecodingquantumldpc,
  title={Layered decoding of quantum {LDPC} codes},
  author={Du Crest, Julien and Garcia-Herrero, Francisco and Mhalla, Mehdi and Savin, Valentin and Valls, Javier},
  booktitle={2023 12th International Symposium on Topics in Coding (ISTC)},
  pages={1--5},
  year={2023},
  organization={IEEE}
}

@article{Panteleev_2021,
   title={Degenerate Quantum {LDPC} Codes With Good Finite Length Performance},
   volume={5},
   ISSN={2521-327X},
   journal={Quantum},
   publisher={Verein zur Forderung des Open Access Publizierens in den Quantenwissenschaften},
   author={Panteleev, Pavel and Kalachev, Gleb},
   year={2021},
   month=nov, pages={585} }

@article{Bravyi_2024,
  title={High-threshold and low-overhead fault-tolerant quantum memory},
  author={Bravyi, Sergey and Cross, Andrew W and Gambetta, Jay M and Maslov, Dmitri and Rall, Patrick and Yoder, Theodore J},
  journal={Nature},
  volume={627},
  number={8005},
  pages={778--782},
  year={2024},
  publisher={Nature Publishing Group UK London}
}

@article{1003839,
  title={Finite-length analysis of low-density parity-check codes on the binary erasure channel},
  author={Di, Changyan and Proietti, David and Telatar, I Emre and Richardson, Thomas J and Urbanke, R{\"u}diger L},
  journal={IEEE Transactions on Information theory},
  volume={48},
  number={6},
  pages={1570--1579},
  year={2002},
  publisher={IEEE}
}

@article{Roffe_2020,
  title={Decoding across the quantum low-density parity-check code landscape},
  author={Roffe, Joschka and White, David R and Burton, Simon and Campbell, Earl},
  journal={Physical Review Research},
  volume={2},
  number={4},
  pages={043423},
  year={2020},
  publisher={APS}
}

@article{Roffe_LDPC_Python_tools_2022,
  title={{LDPC}: Python tools for low density parity check codes},
  author={Roffe, Joschka},
  journal={PyPI https://pypi. org/project/ldpc},
  year={2022}
}

@article{Breuckmann_2021,
  title={Quantum low-density parity-check codes},
  author={Breuckmann, Nikolas P and Eberhardt, Jens Niklas},
  journal={PRX quantum},
  volume={2},
  number={4},
  pages={040101},
  year={2021},
  publisher={APS}
}

@article{vasic2025quantumlowdensityparitycheckcodes,
  title={Quantum Low-Density Parity-Check Codes},
  author={Vasic, Bane and Savin, Valentin and Pacenti, Michele and Borah, Shantom and Raveendran, Nithin},
  journal={arXiv preprint arXiv:2510.14090},
  year={2025}
}

@inproceedings{crest2023stabilizerinactivationmessagepassingdecoding,
  title={Stabilizer inactivation for message-passing decoding of quantum {LDPC} codes},
  author={Du Crest, Julien and Mhalla, Mehdi and Savin, Valentin},
  booktitle={2022 IEEE Information Theory Workshop (ITW)},
  pages={488--493},
  year={2022},
  organization={IEEE}
}

@inproceedings{yao2024beliefpropagationdecodingquantum,
  title={Belief propagation decoding of quantum {LDPC} codes with guided decimation},
  author={Yao, Hanwen and Laban, Waleed Abu and H{\"a}ger, Christian and i Amat, Alexandre Graell and Pfister, Henry D},
  booktitle={2024 IEEE International Symposium on Information Theory (ISIT)},
  pages={2478--2483},
  year={2024},
  organization={IEEE}
}

@article{chytas2025enhancedminsumdecodingquantum,
  title={Enhanced min-sum decoding of quantum codes using previous iteration dynamics},
  author={Chytas, Dimitris and Raveendran, Nithin and Vasic, Bane},
  journal={arXiv preprint arXiv:2501.05021},
  year={2025}
}

@INPROCEEDINGS{Rabeti2025BoundsAN,
  author={Rabeti, Sheida and Moradi, Mohsen and Mahdavifar, Hessam},
  booktitle={2025 IEEE International Symposium on Information Theory (ISIT)}, 
  title={Bounds and New Constructions for Girth-Constrained Regular Bipartite Graphs}, 
  year={2025},
  volume={},
  number={},
  pages={1-6},
  doi={10.1109/ISIT63088.2025.11195214}}

@article{rigby2019modified,
  title={Modified belief propagation decoders for quantum low-density parity-check codes},
  author={Rigby, Alex and Olivier, JC and Jarvis, Peter},
  journal={Physical Review A},
  volume={100},
  number={1},
  pages={012330},
  year={2019},
  publisher={APS}
}

@article{rigby2018augmented,
  title={Augmented decoders for {LDPC} codes},
  author={Rigby, Alex R and Olivier, Jan C and Myburgh, Hermanus C and Xiao, Chengshan and Salmon, Brian P},
  journal={EURASIP Journal on Wireless Communications and Networking},
  volume={2018},
  number={1},
  pages={189},
  year={2018},
  publisher={Springer}
}

@article{wolanski2024ambiguity,
  title={Ambiguity clustering: an accurate and efficient decoder for {QLDPC} codes},
  author={Wolanski, Stasiu and Barber, Ben},
  journal={arXiv preprint arXiv:2406.14527},
  year={2024}
}

@article{hillmann2025localized,
  title={Localized statistics decoding for quantum low-density parity-check codes},
  author={Hillmann, Timo and Berent, Lucas and Quintavalle, Armanda O and Eisert, Jens and Wille, Robert and Roffe, Joschka},
  journal={Nature Communications},
  volume={16},
  number={1},
  pages={8214},
  year={2025},
  publisher={Nature Publishing Group UK London}
}

@article{demarti2024almost,
  title={An almost-linear time decoding algorithm for quantum LDPC codes under circuit-level noise},
  author={deMarti iOlius, Antonio and Etxezarreta Martinez, Imanol and Roffe, Joschka and Etxezarreta Martinez, Josu},
  journal={arXiv e-prints},
  pages={arXiv--2409},
  year={2024}
}

@article{moradi2026sequential,
  title={Sequential {BP}-based Decoding of {QLDPC} Codes},
  author={Moradi, Mohsen and Habib, Salman and Nourozi, Vahid and Mitchell, David GM},
  journal={arXiv preprint arXiv:2602.13420},
  year={2026}
}

\end{document}